\documentclass[10pt]{article}
\usepackage{amsfonts}
\usepackage{amssymb}
\usepackage{amsthm}
\usepackage{amsmath}
\usepackage{authblk} 
\usepackage{graphicx}
\usepackage{hyperref}
\usepackage{enumerate}
\usepackage{tikz}
\usepackage{color}
\usepackage{eucal}

\usepackage{soul,todonotes}

\usepackage{mathrsfs}
\usepackage{hyperref}
\usepackage{datetime}

\DeclareMathAlphabet{\mathpzc}{OT1}{pzc}{m}{it}

\numberwithin{equation}{section}

\newcommand{\ii}{{\rm i}}
\newcommand{\ee}{{\rm e}}
\newcommand{\x}{{\rm x}}

\newcommand{\vol}{{\rm vol}}

\newcommand{\dd}{{\rm d}}
\newcommand{\R}{\mathscr{R}}
\DeclareMathOperator{\supp}{supp}

\newcommand{\beq}{\begin{equation}}
\newcommand{\ene}{\end{equation}}

\newcommand{\g}{{\sf g}}

\newtheorem{thm}{Theorem}
\newtheorem{lemma}[thm]{Lemma}
\newtheorem{prop}[thm]{Proposition}

\theoremstyle{definition}
\newtheorem{defn}[thm]{Definition}
\newtheorem{ex}[thm]{Example}

\newdateformat{daymonthyear}{\THEDAY \, \monthname[\THEMONTH] \THEYEAR}

\begin{document}

\title{Semiclassical gravity with a conformally covariant field in globally hyperbolic spacetimes}
\author{Benito A. Ju\'arez-Aubry$^1${}\thanks{benito.juarez@correo.nucleares.unam.mx}\,}
\author{Sujoy K. Modak$^2${}\thanks{smodak@ucol.mx} }
\affil{$^1$Departamento de Gravitaci\'on y Teor\'ia de Campos, Instituto de Investigaciones Nucleares, Universidad Nacional Aut\'onoma de M\'exico, A. Postal 20-126, CDMX, M\'exico}
\affil{$^2$Facultad de Ciencias, CUICBAS, Universidad de Colima, Colima, C.P. 28045,  M\'exico}
\date{\daymonthyear\today}

\maketitle

\begin{abstract}
%
We prove that semiclassical gravity in conformally static, globally hyperbolic spacetimes with a massless, conformally coupled Klein-Gordon field is well posed, when viewed as a coupled theory for the dynamical conformal factor of the metric and the Klein-Gordon theory. Namely, it admits unique and stable solutions whenever constrained fourth-order initial data for the conformal factor and suitably defined Hadamard initial data for the Klein-Gordon state are provided on a spacelike Cauchy surface. As no spacetime symmetries are imposed on the conformal factor, the present result implies that, provided constrained initial data exists, there also exist exact solutions to the semiclassical gravity equations beyond the isotropic, homogeneous or static cases.
\end{abstract}


\section{Introduction}

The semiclassical gravity equations describe the interacting dynamics between the gravitational field and quantum matter fields. The quantum fields, which propagate in spacetime, gravitate via the expectation value of their stress-energy tensor, $\omega(T_{ab})$, which sources the semiclassical Einstein field equations that describe the dynamics of the metric tensor of spacetime, $g_{ab}$. Namely, we have
\begin{align}
G_{ab} + \Lambda g_{ab} = 8 \pi G_{\rm N} \omega(T_{ab}),
\end{align}
together with the matter fields equations of motion. 

Semiclassical gravity is believed to be relevant as the semiclassical regime of quantum gravity, sufficiently far away from Planck scale, in settings in which both quantum and gravitational effects cannot be neglected, such as cosmology or in black-hole physics. For example, the early stages of black-hole evaporation are believed to be correctly described by semiclassical gravity \cite{Hawking, WaldQFT}. Many aspects of semiclassical gravity are discussed in \cite{Ford}, where the reader can also find a number of references on the subject.

Having said so, the semiclassical gravity equations are complicated and a full mathematical understanding of the theory is lacking. It is currently unknown whether theorems of existence, uniqueness or stability of solutions hold for full semiclassical gravity, even in the cases in which only a free field, such as the Klein-Gordon scalar, appears in the matter sector. The central motivation for this article is advancing our understanding of the mathematical properties of the theory in globally hyperbolic spacetimes in the case in which a massless, conformally coupled Klein-Gordon field appears in the quantum matter sector.

We should mention that steady progress in characterising exact semiclassical gravity solutions has been made since the second half of the last century, especially in cosmological situations. Combinations of numerical and analytic techniques have been used to study cosmological spacetimes in \cite{Starobinsky, Anderson1, Anderson2, Anderson3, Anderson4} and for the particular case of de Sitter spacetime in \cite{Juarez-Aubry:2019jbw}. On the front of exact solutions, results on the existence, uniqueness and stability of solutions in cosmology have been obtained in \cite{Wald:1978ce, Dappiaggi:2008mm, Pinamonti:2010is, Pinamonti:2013wya, Meda:2020smb, Gottschalk:2018kqt}. Semiclassical gravity in static spacetimes has been studied in \cite{Sanders:2020osl, Benito}. 

From a perturbation theory viewpoint, the principle of perturbative agreement for physical solutions in higher-order theories, introduced by Simon in \cite{Simon}, has been applied to semiclassical gravity in \cite{Parker-Simon} and further advocated in \cite{Flanagan-Wald}. \cite{JAMS} studies the existence and uniqueness of solutions from a perturbative point of view as a perturbative initial value problem in a semiclassical scalars toy model inspired by semiclassical gravity. To leading order, the perturbative semiclassical gravity approach corresponds to what is known in the literature as semiclassical backreaction. Interesting and recent applications in the backreaction setting related to strong cosmic censorship appear in \cite{Juarez-Aubry:2015dla, Casals:2016ioo, Casals:2016odj, Casals:2019jfo, Hollands:2019whz, Hollands:2020qpe}. A version of the classical singularity theorems of Hawking and Penrose in semiclassical gravity has been studied very recently in \cite{Fewster:2021mmz}.



In this paper we show that there exist solutions to semiclassical gravity beyond the isotropic, homogeneous or static cases. In particular, we show that in the case of a conformally covariant Klein-Gordon field there exist local space- and time-dependent solutions to the theory consisting of a conformally static spacetime with the conformal Klein-Gordon field in a Hadamard state. This result appears in theorem \ref{MainThm2}. 

The main point is that the theory can be seen as one for the dynamical conformal factor of a conformally-static metric, given appropriate constrained initial data for the conformal factor on a spacelike Cauchy surface. The evolution in a neighbourhood of the Cauchy surface is then shown to be unique and stable in the sense that a small perturbation of data yields a small perturbation for the solution. A key element is that in conformally static spacetimes it is possible to have a {\it bona fide} notion of Hadamard initial data using some of the results in \cite{Benito} for static spacetimes, which allows to prescribe in principle reasonable initial data for semiclassical gravity in this context -- a task that is highly non-trivial in general.

The organisation of the paper is the following. In sec. \ref{sec:prel} we briefly introduce the theory of semiclassical gravity, which also serves the purpose of fixing our notation, including a discussion on the initial data and constraints of the theory. In preparation for the methods required in the proof of theorem \ref{MainThm2}, sec. \ref{sec:Stress-Energy-Conf} contains a number of results for conformally related (quantum and classical) field theories. In particular, it is shown that given a Hadamard state for a Klein-Gordon field in a fixed spacetime, a conformal state defined in a conformally related spacetime is Hadamard for the conformally related Klein-Gordon theory. Thus, knowing the Hadamard singular structure in one of the theories allows one to know the singular structure in the other. Sec. \ref{sec:SemiclassGrav} contains the main result of the article announced above. Final remarks and some perspectives appear in sec. \ref{sec:final}.

\section{Semiclassical gravity preliminaries}
\label{sec:prel}

Here we consider for the matter sector of semiclassical gravity a Klein-Gordon field. The Klein-Gordon algebra in the curved spacetime $(M, g_{ab})$, $\mathscr{A}(M)$, is the unital $\star$-algebra generated by smeared fields of the form $\Phi(f)$, where $f \in C_0^\infty(M)$, which satisfy the following relations: (i) Linearity: $f \mapsto \Phi(f)$ is linear. (ii) Hermiticity: $\Phi(f)^\star = \Phi(\overline{f})$. (iii) Field equation: $\Phi((\Box - m^2 - \xi R) f) = 0$, where $m^2\geq 0$ is the mass parameter and $\xi\in \mathbb{R}$ the curvature coupling of the field. (iv) Commutation relations $[\Phi(f), \Phi(g)] = -\ii (f,\mathcal{E} g) 1\!\!1$, where $g \in C_0^\infty(M)$, $\mathcal{E}$ is the causal propagator of $(\Box - m^2 - \xi R)$, the pairing is in $L^2(M, \dd \vol_g)$ and $1\!\!1$ is the algebra unit element.

States for the Klein-Gordon field are linear maps $\omega: \mathscr{A}(M) \to \mathbb{C}$ that are positive and normalised, i.e., for $a \in \mathscr{A}(M)$, $\omega(a a^\star) \geq 0$ and $\omega(1\!\!1) = 1$. For free theories, the expectation value of the stress-energy tensor can be defined for Hadamard states. These are states for which the integral kernel of the Wightman two-point function, $\omega(\Phi(f) \Phi(g))$, which we denote by $G^+$, takes the following form in a convex normal neighbourhood,
\begin{subequations}
\begin{align}
G^+(\x, \x') & := \omega(\Phi(\x) \Phi(\x')) = H_\ell(\x, \x') +  w_{{\ell} }(\x, \x'),
\label{HadamardCondition} \\
H_\ell(\x, \x') & := \frac{1}{2 (2 \pi)^2} \left( \frac{\Delta^{1/2}(\x, \x')}{\sigma_\epsilon(\x, \x')} +  v(\x, \x') \ln \left(\frac{\sigma_\epsilon(\x, \x')}{\ell^2} \right)  \right) \label{HadamardBiDis}
\end{align}
\end{subequations}
Here $H_\ell$ is known as the Hadamard bi-distribution, $\sigma(\x, \x')$ is half of the squared geodesic distance between the spacetime points $\x$ and $\x'$ and $\sigma_\epsilon(\x, \x') := \sigma(\x, \x') + \ii \epsilon (t(\x)-t(\x') ) + \frac{1}{2}\epsilon^2$ is its regularised version with $t$ an arbitrary time function, $\Delta$ is the van Vleck-Morette determinant, $v$ and $w$ are smooth and symmetric coefficients and $\ell \in \mathbb{R}$ is an arbitrary length scale. The coefficient $v$ admits an asymptotic covariant Taylor series of the form
\begin{align}
v(\x, \x') = \sum_{n = 0}^\infty v_n(\x, \x') \sigma^n(\x, \x'), \label{vHadamard}
\end{align}
where the coefficients $v_n$ obey the so-called Hadamard recursion relations. Note that whenever $G^+$ has Hadamard form its singular structure is fully characterised by $H_\ell$. In particular, $H_\ell$ contains distributional singularities whenever the spacetime points in its argument are connected by a null curve. 

We henceforth consider the arbitrary length-scale $\ell$ to be fixed, and define the expectation value of the renormalised stress-energy tensor in the Hadamard state $\omega$ as 
\begin{subequations}
\label{TabDef}
\begin{align}
\omega &(T_{ab})  := \lim_{\x' \to \x} \Big(\mathcal{T}_{ab} w_\ell (\x, \x') + \frac{1}{ (2 \pi)^2} g_{ab} v_1(\x, \x')\Big) + \alpha_1 g_{ab} + \alpha_2 G_{ab} + \alpha_3 I_{ab} + \alpha_4 J_{ab} , \label{Stress-energyDef} \\
\mathcal{T}_{ab} &  := (1-2\xi ) g_{b}\,^{b'}\nabla_a \nabla_{b'} +\left(2\xi - \frac{1}{2}\right) g_{ab}g^{cd'} \nabla_c \nabla_{d'}  - \frac{1}{2} g_{ab} m^2 + 2\xi \Big[  - g_{a}\,^{a'} g_{b}\,^{b'} \nabla_{a'} \nabla_{b'} + g_{ab} g^{c d}\nabla_c \nabla_d + \frac{1}{2}G_{ab} \Big], \label{TabPointSplit2} \\
I_{ab} & := R_{;ab} - \frac{1}{2} g_{ab} \Box R - \Box R_{ab} + \frac{1}{2} g_{ab} R^{cd} R_{cd} - 2 R^{cd} R_{cdab},
 \label{Iab} \\
J_{ab} & := 2 R_{;ab} - 2 g_{ab} \Box R + \frac{1}{2}g_{ab} R^2 - 2 R R_{ab}, \label{Jab}
\end{align}
\end{subequations}
where $v_1$ is the $n = 1$ Hadamard coefficient in the asymptotic series of $v$, cf. eq. \eqref{vHadamard}, $g_a{}^{b'}(\x, \x')$ is the parallel-transport propagator between the points $\x$ and $\x'$, and $\alpha_i$, $i = 1, \ldots, 4$, are arbitrary real constants that appear as renormalisation ambiguities. 

We shall use extensively Synge's coincidence limit notation, whereby $[A](\x) := \lim_{\x' \to \x} A(\x, \x')$. Assuming that the state $\omega$ is quasi-free, and using eq. \eqref{TabDef}, the semiclassical gravity equations with a Klein-Gordon field read as
\begin{subequations}
\label{semiEFEKG}
\begin{align}
& G_{ab} + \Lambda g_{ab} = 8 \pi G_{\rm N}  [\mathcal{T}_{ab} w_\ell] + \frac{2  G_{\rm N}}{\pi} g_{ab} [v_1] - \alpha I_{ab} +  \beta J_{ab} \label{semiEFE}, \\
& (\Box_\x - m^2 - \xi R(\x)) G^+(\x, \x') = (\Box_{\x'} - m^2 - \xi R(\x')) G^+(\x, \x') = 0, \label{KG}
\end{align}
\end{subequations}
where the cosmological constant, $\Lambda$, and Newton's constant, $G_{\rm N}$, together with $\alpha$ and $\beta$ are redefinitions of the constants $\alpha_i$, ($i = 1, \ldots, 4$) in eq. \eqref{Stress-energyDef} and of some ``bare" cosmological and Newton's constants, and where
\begin{align}
[v_1] = \frac{1}{8} m^4 + \frac{1}{4} \left(\xi - \frac{1}{6}\right) m^2 R - \frac{1}{24}\left(\xi - \frac{1}{5} \right) \Box R   + \frac{1}{8} \left(\xi - \frac{1}{6} \right)^2 R^2 - \frac{1}{720} R_{ab} R^{ab} + \frac{1}{720} R_{abcd} R^{abcd}.
\label{v1} 
\end{align}

\subsection{The initial data of semiclassical gravity}

The problem defined by eq. \eqref{semiEFE} is of fourth order in the spacetime metric. Therefore, the statement of the Cauchy problem for the system \eqref{semiEFEKG} should have as data on the initial surface $\mathcal{C}$ for the spacetime metric
\begin{subequations}
\label{IVP}
\begin{align}
\label{gIVP}
g_{ab}(\x)|_\mathcal{C} = g^{(0)}_{ab}(\underline{x}), & \quad & \dot g_{ab}(\x)|_\mathcal{C} = g^{(1)}_{ab}(\underline{x}), & \quad & 
\ddot g_{ab}(\x)|_\mathcal{C} = g^{(2)}_{ab}(\underline{x}), & \quad & \dddot g_{ab}(\x)|_\mathcal{C} = g^{(3)}_{ab}(\underline{x})
\end{align}
where $\dot A := \partial_t A$ for some time function $t$ and $\underline{\x} \in \mathcal{C}$. The initial data for the state consists of the correlations of the 3-fields, $\varphi$, and 3-momenta, $\pi$, on $\mathcal{C}$ satisfying the canonical commutation relations on $\mathcal{C}$,
\begin{align}
G^+(\x,\x')|_\mathcal{C} = \omega(\varphi(\underline{\x}) \varphi(\underline{\x}')) = G^+_{\varphi\varphi}(\underline{\x}, \underline{\x}'), & \quad & \nabla_n G^+(\x,\x')|_\mathcal{C} = \omega(\pi(\underline{\x}) \varphi(\underline{\x}')) = G^+_{\pi \varphi}(\underline{\x}, \underline{\x}'), \nonumber \\
\nabla_{n'}G^+(\x,\x')|_\mathcal{C} = \omega(\varphi(\underline{\x}) \pi(\underline{\x}')) = G^+_{\varphi \pi}(\underline{\x}, \underline{\x}'), & \quad & \nabla_n \nabla_{n'} G^+(\x,\x')|_\mathcal{C} = \omega(\pi(\underline{\x}) \pi(\underline{\x}')) = G^+_{\pi \pi}(\underline{\x}, \underline{\x}'). \label{WightmanIVP}
\end{align}
\end{subequations}

The initial data \eqref{WightmanIVP} should also satisfy the positivity condition, such that the solution $G^+$ satisfy 
\begin{equation}
\omega( \Phi(f) \Phi(\overline{f})) = \int_{M \times M} \dd \vol(\x) \dd \vol(\x') G^+(\x, \x') f(\x) \overline{f(\x')} \geq 0, 
\label{CovPositivity}
\end{equation}
which imposes on initial data that, given $u, v \in C_0^\infty(\mathcal{C})$, it holds that
\begin{align}
\int_{\mathcal{C} \times \mathcal{C}} \dd \vol_\mathcal{C}(\underline{\x}) \dd \vol_\mathcal{C}(\underline{\x}') \left[ G^+_{\varphi\varphi}(\underline{\x}, \underline{\x}') v(\underline{\x}) \overline{v(\underline{\x}')} - G^+_{\pi\varphi}(\underline{\x}, \underline{\x}') u(\underline{\x}) \overline{v(\underline{\x}')} - G^+_{\varphi\pi}(\underline{\x}, \underline{\x}') v(\underline{\x}) \overline{u(\underline{\x}')} + G^+_{\pi\pi}(\underline{\x}, \underline{\x}') u(\underline{\x}) \overline{u(\underline{\x}')} \right] \geq 0.
\label{Positivity}
\end{align}

For quasi-free states eq. \eqref{Positivity} fully characterises the positivity condition in globally hyperbolic spacetimes.

\begin{prop}
\label{Prop:Pos}
Let $(M, g_{ab})$ be a globally-hyperbolic spacetime, $\mathcal{C} \subset M$ a Cauchy surface. Let $G^+$ be a bi-solution to  the Klein-Gordon equation \eqref{KG} with initial data \eqref{WightmanIVP} defining a quasi-free state. The positivity condition \eqref{CovPositivity} is equivalent to the condition \eqref{Positivity}.
\end{prop}
\begin{proof}
The result follows from the equivalence between covariantly smeared and symplectically smeared two-point functions, see e.g. eq. (3.11) in \cite{JAMS}. Assume \eqref{CovPositivity} holds, then
\begin{align}
0 & \leq \int_{M \times M} \dd \vol(\x) \dd \vol(\x') G^+(\x, \x') f(\x) \overline{f(\x')}  = \int_{\mathcal{C} \times \mathcal{C}} \dd \vol_\mathcal{C}(\underline{\x}) \dd \vol_\mathcal{C}(\underline{\x}') \left[ G^+_{\varphi\varphi}(\underline{\x}, \underline{\x}') [\nabla_n \mathcal{E} f]|_\mathcal{C}(\underline{\x}) \overline{[\nabla_{n'} \mathcal{E} f]|_\mathcal{C}(\underline{\x}')} \right. \nonumber \\
& \left. - G^+_{\pi\varphi}(\underline{\x}, \underline{\x}') [\mathcal{E} f]|_\mathcal{C}(\underline{\x}) \overline{[\nabla_{n'} \mathcal{E} f]|_\mathcal{C}(\underline{\x}')} - G^+_{\varphi\pi}(\underline{\x}, \underline{\x}') [\nabla_n \mathcal{E} f]|_\mathcal{C}(\underline{\x}) \overline{[ \mathcal{E} f]|_\mathcal{C}(\underline{\x}')} + G^+_{\pi\pi}(\underline{\x}, \underline{\x}') [ \mathcal{E} f]|_\mathcal{C}(\underline{\x}) \overline{[ \mathcal{E} f]|_\mathcal{C}(\underline{\x}')} \right],
\end{align}
where $\mathcal{E}: C_0^\infty(M) \to C_{\rm sc}^\infty(M)$ is the causal propagator of $\Box - m^2 - \xi R$. Setting $u = [\mathcal{E} f]|_\mathcal{C}$ and $v = [\nabla_n \mathcal{E} f]|_\mathcal{C}$ one has \eqref{Positivity}.

Assume now eq. \eqref{Positivity} holds. Now $u$ and $v$ in \eqref{Positivity} can be seen as initial data to the classical Klein-Gordon solution $\phi = \mathcal{E}f$. Then \eqref{Positivity} implies \eqref{CovPositivity}.
\end{proof}

If the commutation relations and the positivity condition are satisfied by the solution, the GNS construction yields a concrete representation of the Klein-Gordon theory in Hilbert space from the abstract Klein-Gordon algebra and the two-point function $G^+$.

If the state $\omega$ is quasi-free with non-vanishing one-point function, the Klein-Gordon equation must also be imposed for the one-point function, together with the initial data $\omega(\varphi)$ and $\omega(\pi)$. The case with non-vanishing one-point function is important, for it includes coherent states, for which we expect the semiclassical regime to be appropriate. If the state is not quasi-free, then initial data must be prescribed for all $n$-point functions. Nevertheless, for the purposes of solving the semiclassical gravity equations, relaxing the quasi-free requirement or the vanishing of the one-point function does not add complications to the solvability of the problem -- all $n$-point function equations decouple from eq. \eqref{semiEFE}, except for the two-point function, which is used to define the expectation value of the renormalised stress-energy tensor of the theory. 


As we have mentioned above, it is desirable to consider Hadamard states in the context of semiclassical gravity.\footnote{But see for example \cite{Meda:2020smb} for some semiclassical gravity solutions in cosmology with states that are Hadamard only in an approximate sense.} In this sense, initial data should be such that the resulting states satisfy the Hadamard condition.

\begin{defn}
Let $(M, g_{ab})$ be a smooth, globally hyperbolic spacetime and $\mathcal{C} \subset M$ a Cauchy surface of $(M, g_{ab})$. We say that
\begin{align}
G^+(\x,\x')|_\mathcal{C} = \omega(\varphi(\underline{\x}) \varphi(\underline{\x}')) = G^+_{\varphi\varphi}(\underline{\x}, \underline{\x}'), & \quad & \nabla_n G^+(\x,\x')|_\mathcal{C} = \omega(\pi(\underline{\x}) \varphi(\underline{\x}')) = G^+_{\pi \varphi}(\underline{\x}, \underline{\x}'), \nonumber \\
\nabla_{n'}G^+(\x,\x')|_\mathcal{C} = \omega(\varphi(\underline{\x}) \pi(\underline{\x}')) = G^+_{\varphi \pi}(\underline{\x}, \underline{\x}'), & \quad & \nabla_n \nabla_{n'} G^+(\x,\x')|_\mathcal{C} = \omega(\pi(\underline{\x}) \pi(\underline{\x}')) = G^+_{\pi \pi}(\underline{\x}, \underline{\x}'), 
\end{align}
where $\underline{x}, \underline{x}' \in \mathcal{C}$, is \emph{Hadamard initial data} for the Wightman two-point function of the Klein-Gordon field if the bi-solution to the Klein-Gordon equation, $G^+$, is a Hadamard two-point function.
\end{defn}

\subsection{The constraints of semiclassical gravity and good initial data sets}
\label{subsec:Constraints}

Semiclassical gravity contains four non-dynamical equations analogous to the Gauss and Hamilton constraint equations of General Relativity. For higher order gravity, which can be seen as semiclassical gravity in the $G_{\rm N} \to 0$ limit, the structure of these constraints has been thoroughly studied in \cite{Osorio-Morales}. The constraints -- like in the General Relativity case -- are a consequence of the Gauss-Codazzi equations. If $\mathcal{C}$ is a Cauchy surface of the globally hyperbolic spacetime $(M, g_{ab})$ with normal $n^a$ and induced metric $h_{ab} = g_{ab} - n_a n_b$, then $h^a{}_b = g^{ac} h_{cb}$ defines a projector onto $\mathcal{C}$, and the Gauss-Codazzi equations imply
\begin{subequations}
\begin{align}
h^b{}_a G_{bc} n^c & = D_b K^b{}_a - D_a K^b{}_b, \\
G_{ab} n^a n^b & = \frac{1}{2} \left( {}^{(3)}R + (K^a{}_a)^2 - K_{ab} K^{ab} \right),
\end{align}
\end{subequations}
where $K_{ab}$ is the extrinsic curvature of the Riemannian spacetime $(\mathcal{C}, h_{ab})$ embedded into $(M, g_{ab})$ and $D_a$ is the covariant derivative compatible with $h_{ab}$. Using eq. \eqref{Iab}, \eqref{Jab} and \eqref{v1}, the semiclassical Einstein equations with a conformally-coupled Klein-Gordon field \eqref{semiEFE} can be written as
\begin{align}
H_{ab} & := - \alpha \Box \left(  R_{ab} - \frac{1}{2} R g_{ab}\right) +  (2 \beta - \alpha)(g_{ab} \Box - \nabla_a \nabla_b) R - \frac{ G_{\rm N}}{360 \pi} g_{ab} \Box R -  2 \alpha R^{cd} R_{cdab} + 2 \beta R R_{ab}  \nonumber \\
 & - \frac{1}{2} g_{ab} \left(\frac{ G_{\rm N}}{180 \pi}     R_{cdef} R^{cdef} - \left( \alpha + \frac{ G_{\rm N}}{180 \pi} \right) R^{cd} R_{cd} +  \beta R^2 \right)  +  R_{ab} - \frac{1}{2} R g_{ab} + \Lambda g_{ab}  = 8 \pi G_{\rm N}  [\mathcal{T}_{ab} w_\ell].
 \label{Constr1An}
\end{align}

The highest-order derivatives appear in the first three terms of eq. \eqref{Constr1An}. To see that the normal components of eq. \eqref{Constr1An} contain no fourth order time derivatives of $g_{ab}$, note that the term $n^a n^b \Box \left(  R_{ab} - \frac{1}{2} R g_{ab}\right)$ contains only up to third order time derivatives by the Gauss-Codazzi equations, see eq. (2.9) in \cite{Osorio-Morales}. Similarly, it is easy to see that for the second term of eq. \eqref{Constr1An}, we have that \cite[Eq. (2.10)]{Osorio-Morales}
\begin{align}
n^a n^b (g_{ab} \Box - \nabla_a \nabla_b) R = h^{ab} \nabla_a \nabla_b R,
\end{align}
and hence it contains no fourth order time derivatives. Finally, the third term of eq. \eqref{Constr1An}, which does not appear in higher-order gravity, also contains no fourth order time derivatives. To see this, it suffices to take the trace of eq. \eqref{Constr1An} and note that the right-hand side of
\begin{align}
   \left( 2 \alpha  - 6 \beta + \frac{ G_{\rm N}}{ 90 \pi} \right) \Box R & = -  2 \alpha R^{cd} R_{cda}{}^a  + 2 \beta R^2 - 2 \left(\frac{ G_{\rm N}}{720 \pi}     R_{cdef} R^{cdef} - \left( \alpha + \frac{ G_{\rm N}}{720 \pi} \right) R^{cd} R_{cd}  \right)  -  R  + 4 \Lambda \nonumber \\
 & - 8 \pi G_{\rm N} g^{cd}[\mathcal{T}_{cd} w_\ell] \label{TraceConstr}
\end{align}
contains no fourth order time derivatives.

The spatial-normal projection of \eqref{Constr1An} can be analysed similarly. The Gauss-Codazzi equations guarantee that the term $n^a h^b{}_c \Box \left(  R_{ab} - \frac{1}{2} R g_{ab}\right)$ contains only up to third order time derivatives, see eq. \cite[Eq. (2.12)]{Osorio-Morales}, and it can be easily verified that \cite{Osorio-Morales}
\begin{align}
n^a h^b{}_c (g_{ab} \Box - \nabla_a \nabla_b) R = n^a D_c \nabla_a R.
\end{align}

Thus, the equations
\begin{subequations}
\label{ConstrAn2}
\begin{align}
& n^a n^b \left( H_{ab} - 8 \pi G_{\rm N} [\mathcal{T}_{ab} w_\ell] \right) = 0, \\
& n^a h^b{}_c \left( H_{ab} - 8 \pi G_{\rm N} [\mathcal{T}_{ab} w_\ell] \right) = 0
\end{align}
\end{subequations}
define constraint equations for semiclassical gravity. Given any solution of semiclassical gravity the constraints are preserved by the contracted Bianchi identities,
\begin{align}
\nabla^a \left( H_{ab} - 8 \pi G_{\rm N} [\mathcal{T}_{ab} w_\ell] \right) = 0,
\end{align}
precisely as in the General Relativity case. 

\begin{defn}
\label{def:GoodInitialData}
We say that a data set of the form \eqref{IVP} is a \emph{good initial data set for semiclassical gravity} if it satisfies the constraints \eqref{ConstrAn2} on an initial value surface $\mathcal{C}$ and the data \eqref{WightmanIVP} satisfy the positivity condition on $\mathcal{C}$ as stated in prop. \ref{Prop:Pos}. 
\end{defn}

Note that def. \ref{def:GoodInitialData} does not impose the Hadamard condition for the initial data of the Wightman function \eqref{WightmanIVP}, but the leading divergences of the two-point function should approximate the Hadamard singular structure for the expectation value to exist on the initial-value surface sufficiently well for the expectation value of the stress-energy tensor to exist. 

We note that the construction of good initial data for semiclassical gravity in general is a very technical open question. Indeed, this is a very technical question already in pure General Relativity. Fortunately, we know of several examples of solutions that, when restricted to a Cauchy surface, provide appropriate solutions to the constraints and define good initial data sets for semiclassical gravity \cite{Juarez-Aubry:2019jbw, Wald:1978ce, Dappiaggi:2008mm, Pinamonti:2010is, Pinamonti:2013wya, Meda:2020smb, Gottschalk:2018kqt, Sanders:2020osl, Benito}, which imply that \eqref{ConstrAn2} are not devoid of solutions.

\section{The stress-energy tensor for conformally-related theories}
\label{sec:Stress-Energy-Conf}

Let $(M, {\sf g}_{ab})$ be a smooth, globally hyperbolic spacetime with Riemann curvature ${\sf R}_{abcd}$, Ricci curvature ${\sf R}_{ab}$ and Ricci scalar ${\sf R}$, which is conformally related to $(M, g_{ab})$ through
\begin{align}
g_{ab} = \ee^{2 \theta} {\sf g}_{ab} = \Theta^2 {\sf g}_{ab}, \hspace{1cm} \theta: M \to \mathbb{R},
\label{Conformalg}
\end{align}
where we assume for the moment that $\theta \in C^\infty(M)$, which implies that $(M, g_{ab})$ is also a smooth, globally-hyperbolic spacetime. We denote the Riemann curvature in $(M, g_{ab})$ by $R_{abcd}$ and the Ricci curvature and Ricci scalar by $R_{ab}$ and $R$ respectively. We denote covariant derivatives compatible with ${\sf g}_{ab}$ as $\nabla_a^{({\sf g})}$, and similarly for d'Alambert's operator, whereas covariant derivatives compatible with $g_{ab}$ are label-free.

We consider a classical Klein-Gordon field, $\psi$, in the fixed background $(M, {\sf g}_{ab})$ obeying the field equation
\begin{align}
Q_\g \psi := \left(\Box^{(\g)} - \Theta^2(m^2 + \xi R) - \frac{1}{6} {\sf R} + \frac{1}{6} \Theta^2 R \right) \psi = 0,
\label{Paux1}
\end{align}
where $R$ and $\Theta$ are viewed as fixed spacetime functions in the potential term of eq. \eqref{Paux1}. It is well-known \cite{Fulling} that solutions to \eqref{Paux1} are in one-to-one correspondence to solutions to the problem defined in $(M, g_{ab})$ by
\begin{align}
P_g \phi := \left(\Box - m^2 - \xi R \right) \phi = 0,
\label{Paux2}
\end{align}
through the relation $\phi = \Theta^{-1} \psi$. Moreover, the fundamental Green operators associated to the normally hyperbolic differential operators appearing in eq. \eqref{Paux1} and \eqref{Paux2} are conformally related.

\begin{prop}
\label{PropE}
Let $Q_\g : C^\infty(M) \to C^\infty(M)$ and $P_g : C^\infty(M) \to C^\infty(M)$ be the normally hyperbolic operators defined by $Q_\g$ in \eqref{Paux1} and $P_g$ in \eqref{Paux2} with fundamental retarded ($+$) and advanced ($-$) Green operators $\mathcal{E}^\pm_{Q_\g}: C_0^\infty(M) \to C_{\rm sc}^\infty(M)$ and $\mathcal{E}^\pm_{P_g}: C_0^\infty(M) \to C_{\rm sc}^\infty(M)$, respectively. The integral kernels of the Green operators satisfy distributionally
\begin{align}
E^\pm_{P_g}(\x, \x') = \Theta^{-1}(\x) E^\pm_{Q_\g}(\x, \x') \Theta^{-1}(\x').
\label{ConfE}
\end{align}  
\end{prop}
\begin{proof}
Since we have assumed above that $\Theta = \ee^\theta$ with $\theta \in C^\infty(M)$, $\Theta$ is strictly positive and $\supp \mathcal{E}^\pm_{P_g} = \supp \mathcal{E}^\pm_{Q_\g}$ as distributions. One can verify directly that
\begin{align}
P_g = \Theta^{-3} Q_\g \Theta, 
\end{align}
and using the fact that the volume elements are related by $\dd \vol_g = \Theta^4 \dd \vol_\g$, one has that
\begin{align}
f = \int_M \dd \vol_g (\x') E^\pm_{P_g} (\x, \x') P_g' f(\x') =  \int_M \dd \vol_\g (\x')  E^\pm_{P_g} (\x, \x') \Theta(\x') Q_\g' \Theta(\x') f(\x'),
\label{PropEq1}
\end{align}
where on the left-hand side of eq. \eqref{PropEq1} we have used $f = \mathcal{E}^\pm_{P_g} P_\g f$. Setting $g = \Theta f$ and multiplying eq. \eqref{PropEq1} by $\Theta(\x)$, we have
\begin{align}
 \int_M \dd \vol_\g (\x')  \Theta(\x) E^\pm_{P_g} (\x, \x') \Theta(\x') Q_\g' g(\x') = g(\x) = \int_M \dd \vol_\g (\x')  E^\pm_{Q_\g} (\x, \x') Q_\g' g(\x'),
\end{align}
where on the right-hand side we have used $g = \mathcal{E}^\pm_{Q_\g} Q_\g g$. Since $f$ (and therefore $g$) is arbitrary, and by the uniqueness of the Green fundamental operators, eq. \eqref{ConfE} follows.
\end{proof}

The situation is similar for quantum field theories defined in $(M, \g_{ab})$ and $(M, g_{ab})$. Let $\Psi(f)$ be the generators of the Klein-Gordon algebra $\mathsf{A}(M)$ in the spacetime $(M, \g_{ab})$ and $\Phi(f)$ be the generators of the Klein-Gordon algebra $\mathscr{A}(M)$ in $(M, g_{ab})$. We can define quasi-free states in the theory $\mathsf{A}(M)$ as positive, normalised maps $\omega_\Psi: \mathsf{A}(M) \to \mathbb{C}$ by prescribing the Wightman two-point function
\begin{align}
\omega_\Psi(\Psi(f) \Psi(g)) = \int_{M \times M} \dd \vol(\x) \dd \vol(\x') {\sf G}^+(\x, \x') f(\x) g(\x'),
\label{WightmanConf}
\end{align}
which, together with the one-point function $\omega_\Psi(\Psi(f)) = \psi(f)$ if it is non-vanishing, fully characterises the state of the theory. The kernel of the Wightman function \eqref{WightmanConf} is a distributional bi-solution to the Klein-Gordon equation,
\begin{align}
& \left(\Box^{(\g)}_\x - \Theta^2(\x)(m^2 + \xi R(\x)) - \frac{1}{6} {\sf R}(\x) + \frac{1}{6} \Theta^2(\x) R(\x) \right){\sf G}^+(\x, \x') \nonumber \\
& = \left(\Box^{(\g)}_{\x'} - \Theta^2(\x')(m^2 + \xi R(\x')) - \frac{1}{6} {\sf R}(\x') + \frac{1}{6} \Theta^2(\x') R(\x') \right){\sf G}^+(\x, \x') = 0.
\end{align}

If the one-point function is non-vanishing, then one must impose that it satisfy the Klein-Gordon equation \eqref{Paux1} too. 

It is formally known (see e.g. \cite{Birrell-Davies}) that we can construct (quasi-free) states in the conformally-related theory $\mathscr{A}(M)$ via the formula
\begin{align}
G^+(\x, \x') = \Theta^{-1}(\x) {\sf G}^+(\x, \x') \Theta^{-1}(\x').
\label{ConfWightman}
\end{align}

More precisely, we can show the following result.

\begin{prop}
\label{Prop:Had}
Let $\omega_\Psi: {\sf A}(M) \to \mathbb{C}$ be a Hadamard state with Wightman two-point function kernel ${\sf G}^+$. The kernel \eqref{ConfWightman} defines the Wightman function of a Hadamard state $\omega_\Phi: \mathscr{A}(M) \to \mathbb{C}$.
\end{prop}
\begin{proof}
That $G^+$ is a distributional bi-solution to the Klein-Gordon equation follows from the above discussions. Since $\Theta$ is a positive, real, smooth function, we have that $G^+(\x, \x') = \overline{G^+(\x', \x)}$. The state positivity of $G^+$ follows from the positivity of ${\sf G}^+$ directly. The Hadamard property can be verified as in \cite[Prop. 8]{Benito}, using the fact that the causal propagators of the operators $Q_\g$ (defining the field equation in ${\sf A}(M)$) and of $P_g$ (defining the field equation in $\mathscr{A}(M)$) are related by eq. \eqref{ConfE} of Prop. \ref{PropE}, together with Radzikowski's theorem \cite[Theorem 5.1]{Radzikowski}, in particular the equivalence of items 1 and 3 in that theorem.
\end{proof}

The choice $m^2 = 0$ and $\xi = 1/6$ renders ${\sf A}(M)$ and $\mathscr{A}(M)$ into theories for conformally covariant Klein-Gordon fields. Note that in this case, the stress-energy tensor of the theory ${\sf A}(M)$ is conserved, as the potential term is free of background structure in $(M, {\sf g}_{ab})$, $\Theta$ and $R$. We make this choice now. 

Following eq. (67) of \cite{Decanini-Folacci}, we can write the expectation value of the stress-energy tensor \eqref{Stress-energyDef} for the conformal Klein-Gordon field $\Phi$ in the state $\omega_\Phi$ as
\begin{align}
\omega_\Phi(T_{ab}) & = \frac{1}{2(2 \pi)^2}\left(- [w_{ab}] + \frac{1}{2} g_{ab} [\Box w] + \frac{1}{3} [w]_{;ab} - \frac{1}{12}  g_{ab} \Box [w] + \frac{1}{6} \left( R_{ab} - \frac{1}{2} g_{ab} R \right)[w] \right) \nonumber \\
& + \frac{1}{(2\pi)^2} g_{ab} [v_1] + \alpha_1 g_{ab} + \alpha_2 G_{ab} + \alpha_3 I_{ab} + \alpha_4 J_{ab},
\label{TabDeca}
\end{align}
where semicolons denote covariant derivation with respect to the connection $\nabla_a$ and with
\begin{subequations}
\begin{align}
[w](\x) & := \lim_{\x'\to \x} G^+(\x, \x') - H(\x, \x'), \\
[w_{ab}](\x) & := \lim_{\x'\to \x} \nabla_b \nabla_a \left[G^+(\x, \x') - H(\x, \x') \right], \\
[\Box w](\x) & := \lim_{\x'\to \x} g^{ab} \nabla_b \nabla_a \left[G^+(\x, \x') - H(\x, \x') \right] = g^{ab}[w_{ab}](\x).
\end{align}
\end{subequations}

A similar expression can be written for $\omega_\Psi(T_{ab})$. Using Prop. \eqref{Prop:Had}, eq. \eqref{TabDeca} can be written in terms of the limits of the regular part of the two-point function ${\sf G}^+$ and its derivatives,
\begin{subequations}
\begin{align}
[{\sf w}](\x) & := \lim_{\x'\to \x} {\sf G}^+(\x, \x') - {\sf H}(\x, \x'), \\
[{\sf w}_{a}](\x) & := \lim_{\x'\to \x}  \nabla_a \left[{\sf G}^+(\x, \x') - {\sf H}(\x, \x') \right], \\
[{\sf w}_{ab}](\x) & := \lim_{\x'\to \x} \nabla_b \nabla_a \left[{\sf G}^+(\x, \x') - {\sf H}(\x, \x') \right], \label{sfwab}
\end{align}
\end{subequations}
and the conformal factor and its covariant derivatives. Indeed, we have that
\begin{subequations}
\label{w-conf}
\begin{align}
[w] & 
= \Theta^{-2} [w] = \ee^{-2 \theta} [{\sf w}], \\
[w_{ab}] 
& = \Theta^{-1}   \left(\Theta^{-1}_{;ab} [{\sf w}]  + \Theta^{-1}_{;a} [{\sf w}_{b}] + \Theta^{-1}_{;b} [{\sf w}_{a}]+ \Theta^{-1} [{\sf w}_{ab}]  \right) \nonumber \\
& = \ee^{-2 \theta} \left( \left(\theta_{;a} \theta_{;b} - \theta_{;ab} \right) [{\sf w}]  - \theta_{;a} [{\sf w}_{b}] - \theta_{;b} [{\sf w}_{a}]+  [{\sf w}_{ab}] \right).
\end{align}
\end{subequations}

Using eq. \eqref{w-conf}, we can write eq. \eqref{TabDeca} as
\begin{subequations}
\label{TabConformal}
\begin{align}
\omega_\Phi(T_{ab}) & = [\mathcal{T}_{ab} w] + \frac{1}{(2\pi)^2} g_{ab} [v_1] + \alpha_1 g_{ab} + \alpha_2 G_{ab} + \alpha_3 I_{ab} + \alpha_4 J_{ab}, \\
[\mathcal{T}_{ab} w] & = \frac{\ee^{2\theta}}{2(2 \pi)^2}\left(  \frac{1}{3} \left(\theta_{;a} \theta_{;b} + \theta_{;ab} \right) [{\sf w}]  + \theta_{;a} [{\sf w}_{b}] + \theta_{;b} [{\sf w}_{a}] -  [{\sf w}_{ab}]  +\frac{1}{3} \left(  - 2 \theta_{;a} [{\sf w}]_{;b} - 2 \theta_{;b} [{\sf w}]_{;a} + [{\sf w}]_{;ab} \right)  \right. \nonumber \\
& \left. + \frac{1}{2} g_{ab} g^{cd}\left( \left(\frac{1}{3}\theta_{;c} \theta_{;d} - \frac{2}{3}\theta_{;cd} \right) [{\sf w}]  - \theta_{;c} [{\sf w}_{d}] - \theta_{;d} [{\sf w}_{c}]+  [{\sf w}_{cd}] -\frac{1}{6} \left( - 2 \theta_{;c} [{\sf w}]_{;d} - 2 \theta_{;d} [{\sf w}]_{;c} + [{\sf w}]_{;cd} \right) \right) \right. \nonumber \\
& \left. + \frac{1}{6} \left( R_{ab} - \frac{1}{2} g_{ab} R \right)[{\sf w}] \right). \label{TabConfw}
\end{align}
\end{subequations}

The form of $\omega_\Phi(T_{ab})$ given by eq. \eqref{TabConformal} in terms of the coincidence limit of the regular part of the state with two-point function ${\sf G}^+$ (and its derivatives) and of the function $\theta$ appearing in the conformal factor will be useful in order to construct the semiclassical gravity solutions with a conformally coupled field.

\section{Semiclassical gravity with a conformally covariant field}
\label{sec:SemiclassGrav}

Using the explicit form of eq. \eqref{Iab}, \eqref{Jab} and \eqref{v1}, the semiclassical Einstein equations with a conformally-coupled Klein-Gordon field \eqref{semiEFE} can be written as
\begin{align}
 (\alpha - 2 \beta) R_{;ab} - \alpha \Box R_{ab} - \left( \frac{\alpha}{2}  - 2 \beta + \frac{ G_{\rm N}}{360 \pi} \right) g_{ab} \Box R -  2 \alpha R^{cd} R_{cdab} + 2 \beta R R_{ab} & \nonumber \\
 - \frac{1}{2} g_{ab} \left(\frac{ G_{\rm N}}{180 \pi}     R_{cdef} R^{cdef} - \left( \alpha + \frac{ G_{\rm N}}{180 \pi} \right) R^{cd} R_{cd} +  \beta R^2 \right)  +  R_{ab} - \frac{1}{2} R g_{ab} + \Lambda g_{ab} & = 8 \pi G_{\rm N}  [\mathcal{T}_{ab} w_\ell].
 \label{SGEFE1}
\end{align}

Note that the system of eq. \eqref{SGEFE1} is closely related to the higher-order gravity theory studied in \cite{Stelle, Noakes} and recently in \cite{Osorio-Morales}. In particular, note that eq. \eqref{SGEFE1} reduce to eq. (5) in \cite{Noakes} if one sets $G_{\rm N} = 0$ and $\Lambda = 0$.\footnote{Our definitions of $\alpha$ and $\beta$ differ from those in \cite{Noakes} by a factor of $16 \pi G_{\rm N}$ and our curvature tensor definitions are also different, resulting in some relative minus signs in the quadratic curvature terms. Noakes uses a convention in \cite{Noakes} consistent with that of the seminal monograph of Birrell and Davies \cite{Birrell-Davies}, while our convention is consistent with the one adopted by Decanini and Folacci in \cite{Decanini-Folacci}.}

Following \cite{Noakes}, it is convenient to split eq. \eqref{SGEFE1} into trace and traceless equations, 
\begin{subequations}
\label{ClassicalSystem}
\begin{align}
&  -  \left( 2 \alpha  - 6 \beta + \frac{ G_{\rm N}}{ 90 \pi} \right) \Box R -  2 \alpha R^{cd} R_{cda}{}^a  + 2 \beta R^2 - 2 \left(\frac{ G_{\rm N}}{720 \pi}     R_{cdef} R^{cdef} - \left( \alpha + \frac{ G_{\rm N}}{720 \pi} \right) R^{cd} R_{cd}  \right)  -  R  + 4 \Lambda \nonumber \\
& = 8 \pi G_{\rm N} g^{cd}[\mathcal{T}_{cd} w_\ell], \label{TraceEq} \\
& (\alpha - 2 \beta) \left( R_{;ab} - \frac{1}{4}  g_{ab} \Box R \right) - \alpha \Box \tilde{R}_{ab}  -  2 \alpha R^{cd} \left(R_{cdab} - \frac{1}{4} g_{ab} R_{cde}{}^e \right) + 2 \beta R \tilde{R}_{ab}  +  \tilde{R}_{ab}  \nonumber \\
& = 8 \pi G_{\rm N} \left(   [\mathcal{T}_{ab} w_\ell]- \frac{1}{4} g_{ab} g^{cd}[\mathcal{T}_{cd} w_\ell] \right),
\label{TraceFree}
\end{align}
\end{subequations}
where $\tilde R_{ab} := R_{ab} - (1/4) R g_{ab}$ is the traceless part of the Ricci scalar. The term $\Box R$ on the left-hand side of eq. \eqref{TraceFree} can be eliminated using eq. \eqref{TraceEq}, and appearances of the (traceful) Ricci tensor in eq. \eqref{TraceEq} are short-hand notation for $R_{ab} = \tilde R_{ab} + (1/4) R g_{ab}$.

It is known how to solve the system \eqref{ClassicalSystem} in the case in which $G_{\rm N} = 0$ as an initial-value problem \cite[Sec. IV]{Noakes}. The point is to introduce in addition to $g_{ab}$ the auxiliary variables $\mathscr{R}$ and $\tilde{\mathscr{R}}_{ab}$, which at the level of solutions should coincide with the Ricci scalar and the traceless part of the Ricci tensor, $R$ and $R_{ab} - (1/4) R g_{ab}$ respectively, by imposing appropriate constraints that should hold everywhere for the physical solution. The system \eqref{ClassicalSystem}, together with the harmonic gauge condition for $g_{ab}$ and the constraints, can be enlarged into a quasi-linear, diagonal, second order, hyperbolic system, and Leray's theorem \cite{Leray} guarantees that given second order, smooth initial data for $g_{ab}$ $\mathscr{R}$ and $\tilde{\mathscr{R}}_{ab}$ on $\mathcal{C}$ satisfying the constraints, the problem admits a unique, smooth solution in a neighbourhood of $\mathcal{C}$. Furthermore, perturbations of initial data lead to perturbative effects in the solutions.

The purpose of this section is to extend these existence, uniqueness and stability results to semiclassical gravity in the case in which one deals with a conformally covariant field in conformally static, globally hyperbolic spacetimes. The first step will be to prescribe criteria for Hadamard initial data for the problem.

\subsection{Hadamard initial data}

It is known from \cite{Benito} how to prescribe Hadamard initial data for static spacetimes. The following results shows that given Hadamard initial data in a static spacetime, it is possible to prescribe Hadamard initial data for a conformally covariant field in a conformally static spacetime. This is nothing but the initial data for conformal states.

\begin{lemma}
\label{LemHadInDat}
Let $(M, {\sf g}_{ab})$ be a fixed smooth, static, globally hyperbolic spacetime and $\mathcal{C} \subset M$ a Cauchy surface of $(M, {\sf g}_{ab})$ with normal ${\sf n}^a$. Let $N_\mathcal{C} \subset M$ be a neighbourhood of $\mathcal{C}$, $\theta \in C^\infty(N_\mathcal{C}, \mathbb{R})$ and $(N_\mathcal{C}, g_{ab} = \ee^{2 \theta} {\sf g}_{ab})$ a (not necessarily static) globally hyperbolic spacetime. Let
\begin{align}
G^+(\x,\x')|_\mathcal{C} = \omega(\varphi(\underline{\x}) \varphi(\underline{\x}')) = G^+_{\varphi\varphi}(\underline{\x}, \underline{\x}'), & \quad & \nabla_n G^+(\x,\x')|_\mathcal{C} = \omega(\pi(\underline{\x}) \varphi(\underline{\x}')) = G^+_{\pi \varphi}(\underline{\x}, \underline{\x}'), \nonumber \\
\nabla_{n'}G^+(\x,\x')|_\mathcal{C} = \omega(\varphi(\underline{\x}) \pi(\underline{\x}')) = G^+_{\varphi \pi}(\underline{\x}, \underline{\x}'), & \quad & \nabla_n \nabla_{n'} G^+(\x,\x')|_\mathcal{C} = \omega(\pi(\underline{\x}) \pi(\underline{\x}')) = G^+_{\pi \pi}(\underline{\x}, \underline{\x}'), 
\end{align}
where $\underline{x}, \underline{x}' \in \mathcal{C}$, be initial data for the two-point function $G^+$ of a conformally covariant Klein-Gordon field in $(N_\mathcal{C}, g_{ab})$, where the covariant derivatives $\nabla_n$ are taken in the direction of $n^a := \ee^{-\theta} {\sf n}^a$. If it holds that
\begin{align}
 G^+_{\varphi\varphi}(\underline{\x}, \underline{\x}') & = {\sf G}^+_{\varphi \varphi}(\underline{\x}, \underline{\x}') \left.\left[\ee^{-\theta(\x)}  \ee^{-\theta(\x')}\right]\right|_\mathcal{C}, \nonumber \\
  G^+_{\pi \varphi}(\underline{\x}, \underline{\x}') & = {\sf G}^+_{\pi \varphi}(\underline{\x}, \underline{\x}') \left.\left[\ee^{-2 \theta(\x)}  \ee^{-\theta(\x')}\right]\right|_\mathcal{C} + {\sf G}^+_{\varphi \varphi}(\underline{\x}, \underline{\x}') \nabla_n \left.\left[\ee^{-\theta(\x)}  \ee^{-\theta(\x')}\right]\right|_\mathcal{C}, \nonumber \\
 G^+_{\varphi \pi}(\underline{\x}, \underline{\x}') & = {\sf G}^+_{\varphi \pi}(\underline{\x}, \underline{\x}') \left.\left[\ee^{-\theta(\x)}  \ee^{-2\theta(\x')}\right]\right|_\mathcal{C} + {\sf G}^+_{\varphi \varphi}(\underline{\x}, \underline{\x}') \nabla_{n'} \left.\left[\ee^{-\theta(\x)}  \ee^{-\theta(\x')}\right]\right|_\mathcal{C}, \nonumber \\
  G^+_{\pi \pi}(\underline{\x}, \underline{\x}') & = {\sf G}^+_{\pi \pi}(\underline{\x}, \underline{\x}') \left.\left[\ee^{-2\theta(\x)}  \ee^{-2\theta(\x')}\right]\right|_\mathcal{C} + {\sf G}^+_{\varphi \pi}(\underline{\x}, \underline{\x}') \nabla_n \left.\left[\ee^{-\theta(\x)}  \ee^{-2\theta(\x')}\right]\right|_\mathcal{C} \nonumber \\
  & + {\sf G}^+_{\pi \varphi}(\underline{\x}, \underline{\x}') \nabla_{n'} \left.\left[\ee^{-2\theta(\x)}  \ee^{-\theta(\x')}\right]\right|_\mathcal{C} + {\sf G}^+_{\varphi \varphi}(\underline{\x}, \underline{\x}') \nabla_n \nabla_{n'} \left.\left[\ee^{-\theta(\x)}  \ee^{-\theta(\x')}\right]\right|_\mathcal{C}, 
\label{HadInData}
\end{align}
where 
\begin{align}
{\sf G}^+(\x,\x')|_\mathcal{C} = {\sf G}^+_{\varphi \varphi}(\underline{\x}, \underline{\x}'), & \quad & \nabla_{\sf n} {\sf G}^+(\x,\x')|_\mathcal{C} = {\sf G}^+_{\pi \varphi}(\underline{\x}, \underline{\x}'), \nonumber \\
\nabla_{{\sf n}'}{\sf G}^+(\x,\x')|_\mathcal{C} = {\sf G}^+_{\varphi \pi}(\underline{\x}, \underline{\x}'), & \quad & \nabla_{\sf n} \nabla_{{\sf n}'} {\sf G}^+(\x,\x')|_\mathcal{C} = {\sf G}^+_{\pi \pi}(\underline{\x}, \underline{\x}')
\label{StatHadInData}
\end{align}
is Hadamard initial data for the two-point function of a conformally covariant Klein-Gordon field in the static spacetime $(M, {\sf g}_{ab})$, then \eqref{HadInData} is Hadamard initial data.
\end{lemma}
\begin{proof}
The Hadamard Wightman function ${\sf G}^+$ is the (unique) solution to
\begin{align}
\left(\Box^{(\g)}_\x + \frac{1}{6} {\sf R}(\x) \right) {\sf G}^+(\x, \x') = \left(\Box^{(\g)}_{\x'} + \frac{1}{6} {\sf R}(\x') \right) {\sf G}^+(\x, \x') = 0
\label{sfGSolve}
\end{align}
with initial data \eqref{StatHadInData}.

It can be explicitly obtained using the causal propagator of the operator $\Box^{(\g)} - \frac{1}{6} {\sf R}$ using eq. (3.19) in \cite{JAMS}.  The Wightman function $G^+(\x, \x') = \ee^{-\theta(\x)} {\sf G}^+(\x, \x') \ee^{-\theta(\x')}$ uniquely satisfies
\begin{align}
\left(\Box_\x - \frac{1}{6} R(\x) \right) G^+(\x, \x') = \left(\Box_{\x'} - \frac{1}{6} R(\x') \right) G^+(\x, \x') = 0.
\end{align}
with initial data \eqref{HadInData}. Prop. \ref{Prop:Had} implies that $G^+$ has Hadamard form and is positive. Thus, eq. \eqref{HadInData} is Hadamard initial data.
\end{proof}

\subsection{Well-posedness of semiclassical gravity for conformally covariant fields in conformally static spacetimes}

Lemma \ref{LemHadInDat} allows us to now turn to the main result of this paper, which is the well-posedness of semiclassical gravity for a conformally covariant Klein-Gordon field in conformally static spacetimes.

\begin{thm}
\label{MainThm2}
Let $(M, {\sf g}_{ab})$ be a fixed smooth, static, globally hyperbolic spacetime with a global, non-vanishing, irrotational timelike Killing vector field $K^a = \partial_t^a$, 
and let $K^a$ be irrotational with respect to $\mathcal{C} \subset M$, a Cauchy surface of $(M, {\sf g}_{ab})$. Given a good initial data set \eqref{IVP} for the semiclassical gravity equations \eqref{semiEFEKG} with a conformally coupled Klein-Gordon field ($m^2 = 0$, $\xi = 1/6$) (in the sense of def. \eqref{def:GoodInitialData}), where the metric data is given in terms of the functions $\theta_0, \theta_1, \theta_2, \theta_3 \in C^\infty(\mathcal{C}, \mathbb{R})$ as 
\begin{subequations}
\label{gCauchyData}
\begin{align}
g_{ab}^{(0)} & = \ee^{2\theta_0} [ {\sf g}_{ab}]|_\mathcal{C}, \\ 
g_{ab}^{(1)}  & =  2 \theta_1 \ee^{2 \theta_0} [\g_{ab}]|_\mathcal{C}, \\
g_{ab}^{(2)} & = 2\left(2 \theta_1^2 + \theta_2\right) \ee^{2 \theta_0} [\g_{ab}]|_{\mathcal{C}}, \\ 
g_{ab}^{(3)} & = 2\left(4 \theta_1^3 + 6 \theta_1 \theta_2 + \theta_3 \right) \ee^{2 \theta_0} [\g_{ab}]|_{\mathcal{C}},
\end{align}
\end{subequations}
and where the where the Wightman two-point function data is Hadamard initial data constructed as in lemma \ref{LemHadInDat}, eq. \eqref{HadInData}, using the initial data $\theta_0$ and $\theta_1$ on $\mathcal{C}$ and data for a conformally related Klein-Gordon two-point function ${\sf G}^+$ in the static spacetime, there exists a neighbourhood of $\mathcal{C}$, $N_\mathcal{C} \subset M$, for which there is a unique function $\theta: N_\mathcal{C} \to \mathbb{R}$, such that the metric tensor
\begin{align}
g_{ab}  = \Theta^2 {\sf g}_{ab} = \ee^{2 \theta} {\sf g}_{ab}, 
\label{Conformalgthm2}
\end{align}
and the unique Wightman two-point function, $G^+$, obtained from the Hadamard initial data yield a solution to semiclassical gravity \eqref{semiEFEKG} in the globally hyperbolic spacetime $(N_\mathcal{C}, g_{ab})$ for generic values of $\alpha$ and $\beta$ (i.e., almost everywhere in the parameter space $\mathbb{R}^2$ endowed with its standard measurable space structure). Moreover, small perturbations of good initial data in the conformal class of $\g_{ab}$ lead to small perturbations in the solution.

\end{thm}
\begin{proof}
As seen in lemma \ref{LemHadInDat}, using the functions $\theta_0$ and $\theta_1$, there is a one-to-one correspondence between the Hadamard initial data of $G^+$ and that of a Wightman function ${\sf G}^+$ that obeys eq. \eqref{sfGSolve}, which is of Hadamard form in the static spacetime $(M, \g_{ab})$. There is furthermore a one-to-one relation between $G^+$ as a state in a neighbourhood of the Cauchy surface, $N_\mathcal{C}$, and ${\sf G}^+$ restricted (more precisely, pulled back) to the same neighbourhood, provided that $\theta$ can be uniquely determined in $N_\mathcal{C}$. Moreover, in this case $G^+$ will be of Hadamard form by prop. \ref{Prop:Had}.

Thus, the goal of the proof is to show that such unique $\theta:N_\mathcal{C} \to \mathbb{R}$ exists by the semiclassical Einstein equations, given the fourth-order initial data $\theta_0, \ldots, \theta_3 \in C^\infty(\mathcal{C}, \mathbb{R})$.

To this end, following \cite{Noakes}, we split eq. \eqref{semiEFE} into trace and traceless equations, cf. eq. \eqref{ClassicalSystem}, and introduce the variables $\mathscr{R}$ and $\tilde{\mathscr{R}}_{ab}$ as discussed above, whereby by setting $g_{ab} = \ee^{2 \theta} \g_{ab}$ we write the system \eqref{semiEFE} as
\begin{subequations}
\label{TraceTracelessThm}
\begin{align}
&  \Box \mathscr{R}  - \left( 2 \alpha  - 6 \beta + \frac{ G_{\rm N}}{90 \pi} \right)^{-1} \left[ -  2 \alpha \left( \tilde{\R}^{cd} + \frac{1}{4} \R \ee^{-2 \theta}{\sf g}^{cd} \right) R_{cda}{}^a   +  2 \beta \R^2  -\frac{ G_{\rm N}}{360 \pi}     R_{cdef} R^{cdef} \right. \nonumber \\
 & \left. + \left( 2 \alpha + \frac{ G_{\rm N}}{360 \pi} \right) \left( \tilde{\R}^{cd} + \frac{1}{4} \R \ee^{-2 \theta} {\sf g}^{cd} \right) \left( \tilde{\R}_{cd} + \frac{1}{4} \R \ee^{2 \theta} {\sf g}_{cd} \right)  -  \R  + 4 \Lambda \right] = 0, \label{TraceEqThm} \\
& - \alpha \Box \tilde{\R}_{ab} + (\alpha - 2 \beta) \R_{;ab}   -  2 \alpha \left( \tilde{\R}^{cd} + \frac{1}{4} \R \ee^{-2 \theta} {\sf g}^{cd} \right) \left(R_{cdab} - \frac{1}{4} \ee^{2 \theta} {\sf g}_{ab} R_{cde}{}^e \right) + 2 \beta \R \tilde{\R}_{ab}  +  \tilde{\R}_{ab}  \nonumber \\
& - \frac{1}{4} (\alpha - 2 \beta)    \ee^{2 \theta} {\sf g}_{ab} \left( 2 \alpha  - 6 \beta + \frac{ G_{\rm N}}{360 \pi} \right)^{-1} \left[ -  2 \alpha \left( \tilde{\R}^{cd} + \frac{1}{4} \R \ee^{- 2 \theta} {\sf g}^{cd} \right) R_{cda}{}^a -\frac{ G_{\rm N}}{360 \pi}     R_{cdef} R^{cdef} \right. \nonumber \\
 & \left. + \left( 2 \alpha + \frac{ G_{\rm N}}{360 \pi} \right) \left( \tilde{\R}^{cd} + \frac{1}{4} \R \ee^{-2 \theta} {\sf g}^{cd} \right) \left( \tilde{\R}_{cd} + \frac{1}{4} \R \ee^{2 \theta} {\sf g}_{cd} \right)  -  \R  + 4 \Lambda \right] - 8 \pi G_{\rm N} \left(   [\mathcal{T}_{ab} w_\ell]- \frac{1}{4} {\sf g}_{ab} {\sf g}^{cd}[\mathcal{T}_{cd} w_\ell] \right) = 0. \label{TracelessEqThm}
\end{align}

In eq. \eqref{TraceEq} and \eqref{TracelessEqThm} the Riemann tensor and its contractions are seen as functionals of $\theta$, $\nabla_a^{(\sf g)} \theta$ and $\nabla_{a}^{(\sf g)} \nabla_{b}^{(\sf g)} \theta$ and of the background fields ${\sf g}_{ab}$ and ${\sf R}_{abcd}$. See e.g. eq. (D.7) in \cite{WaldGR}.

The second order initial data of $\R$ and $\tilde{\R}_{ab}$ will be constrained to match the fourth order initial data of $\theta$, such that for the solutions $\R = R$ and $\tilde{\R}_{ab} = \tilde{R}_{ab}$. 
Note that eq. \eqref{TraceEqThm} is independent of the state -- the only contribution to the stress-energy tensor comes from the trace anomaly and the geometric ambiguities proportional to $\alpha$ and $\beta$. Using formula \eqref{TabConfw} the last term on the left-hand side of eq. \eqref{TraceEqThm} can be seen as a functional of $\theta$, $\nabla_a^{(\sf g)} \theta$ and $\nabla_{a}^{(\sf g)} \nabla_{b}^{(\sf g)} \theta$ and the background fields ${\sf g}_{ab}$, its derivatives, $[{\sf w}]$, $[{\sf w}_a]$ and $[{\sf w}_{ab}]$ together with their covariant derivatives, which are computed from the known Wightman function ${\sf G}^+$ in static spacetime. In addition to eq. \eqref{TraceEqThm} and \eqref{TracelessEqThm}, $\theta$ must satisfy the equation
\begin{align}
 6 \Box^{(\sf g)} \theta +  6 {\sf g}^{ab} \nabla_a^{(\sf g)} \theta \nabla_b^{(\sf g)} \theta + \ee^{2 \theta} \R - {\sf R} = 0,   
 \label{BoxThetaeq}
\end{align}
\end{subequations}
which connects the Ricci scalars in the two conformally related spacetimes whenever $\R = R$.

The system of equations \eqref{TraceTracelessThm} can be cast in the form of a quasilinear, diagonal, second-order hyperbolic one adapting the techniques introduced in \cite{Noakes}. We define $v_{a} := \nabla_a^{(\sf g)} \theta$ and $w_{ab} := \nabla_a^{(\sf g)} \nabla_a^{(\sf g)} \theta$. These variables obey the once- and twice-covariantly-differentiated eq. \eqref{BoxThetaeq} respectively. Defining further $\mathscr{V}_a := \nabla_a^{(\g)} \R$, which obeys the covariantly differentiated eq. \eqref{TraceEqThm}, the system of eq. \eqref{TraceTracelessThm} can be extended into the system
\begin{subequations}
\label{ExtendedSystem}
\begin{align}
   \Box^{({\sf g})} \mathscr{R}  + \mathscr{A} &= 0, \\
 \Box^{({\sf g})} \mathscr{V}_a  + \mathscr{A}_a &= 0, \\
  \Box^{({\sf g})} \tilde{\R}_{ab} + \mathscr{B}_{ab} &= 0, \\
 6 \Box^{(\sf g)} \theta +  \mathscr{C} &= 0, \\
 6 \Box^{(\sf g)} v_a +  \mathscr{C}_a &= 0, \\
 6 \Box^{(\sf g)} w_{ab} +  \mathscr{C}_{ab} &= 0,
 \end{align}
\end{subequations}
where $\mathscr{A}$, $\mathscr{A}_a$, $\mathscr{B}_{ab}$, $\mathscr{C}$, $\mathscr{C}_a$ and $\mathscr{C}_{ab}$ depend at most on $\theta$, $v_a$, $w_{ab}$, $\tilde{\R}$, $\tilde{\R}_{ab}$, $\mathscr{V}_a$ and up to their first order derivatives in time, and on the background structure ${\sf g}_{ab}$, $[{\sf w}]$, $[{\sf w}_a]$, $[{\sf w}_{ab}]$ and their derivatives. The details of the system \eqref{ExtendedSystem} appear in app. \ref{app}. Leray's theorem \cite{Leray} guarantees that the system \eqref{ExtendedSystem} is well posed given initial data and that small perturbations of initial data lead to small perturbations in the solutions. 

All that is left is to constraint $\R$ and $\tilde{\R}_{ab}$ such that they coincide with $R$ and $R_{ab} - (1/4) R g_{ab}$ (resp.), as computed from $g_{ab} = \ee^{2 \theta} {\sf g}_{ab}$. Noakes has shown how to guarantee that this holds in the harmonic gauge for higher-derivative gravity \cite{Noakes}, and the case at hand is a straightforward adaptation of that argument. The symmetric tensor
\begin{align}
\Delta_{ab} := R_{ab} - \frac{1}{2} R \ee^{2 \theta} \g_{ab} - \tilde{\R}_{ab} + \frac{1}{4} \ee^{2 \theta} \g_{ab} \R,
\end{align}
should vanish for physical solutions. Introducing (locally) the harmonic potential,
\begin{align}
F^{\mu} := 
- \frac{1}{2} \ee^{-4 \theta} {\sf g}^{\alpha \beta} \g^{\mu \gamma} \left(\partial_\alpha  (\ee^{2 \theta} \g_{\beta \gamma}) + \partial_\beta (\ee^{2 \theta} \g_{\alpha \gamma}) - \partial_\gamma (\ee^{2 \theta}\g_{\alpha \beta}) \right),
\label{HarmonicPotential}
\end{align}
where Greek indices denote coordinates (as opposed to abstract indices), it holds in harmonic coordinates that \cite[Eq. (45)]{Noakes}
\begin{align}
\frac{1}{2} \ee^{-2 \theta} \g^{\alpha \beta} \partial_\alpha \partial_\beta F^{\mu} + p^\mu - \ee^{-2 \theta} \g^{\mu \alpha} \Delta_{\alpha}{}^\beta{}_{;\beta} = 0,
\label{FmuConstr}
\end{align}
where $p^\mu$ depends in up to first order derivatives of the variables, and
\begin{align}
g^{\gamma \delta} \partial_\gamma \partial_\delta \Delta_{\alpha}{}^\beta{}_{;\beta} + q_\alpha + L_\alpha = 0,
\label{DeltaConstr}
\end{align}
where $L_\alpha$ also depends in up to first order derivatives of the variables and $q_\alpha$ up to first order derivatives in all variables, except for $F^\mu$, including terms of the form $F^\mu{}_{;\alpha \beta}$. 

The system \eqref{FmuConstr} and \eqref{DeltaConstr} for $\Delta_{\alpha}{}^\beta{}_{;\beta}$ and $F^\mu$ can be enlarged into a second order, quasi-linear, diagonal hyperbolic system by introducing $G^\mu_\nu := \partial_\nu F^\mu{}$, which obeys the differentiated version of eq. \eqref{FmuConstr} as field equation. Thus, by Leray's theorem, choosing vanishing initial data for $F^\mu$ and $\Delta_{\alpha}{}^\beta{}_{;\beta}$, eq. \eqref{FmuConstr} and \eqref{DeltaConstr} guarantee that $\R = R$ and $\tilde{\R}_{\mu \nu} = R_{\mu \nu} - (1/4) R g_{\mu \nu}$ throughout the solution. The vanishing initial data for eq. \eqref{DeltaConstr}, $\Delta_{\alpha}{}^\beta{}_{;\beta} = \dot{\Delta}_{\alpha}{}^\beta{}_{;\beta} = 0$ imposes eight constraints, which are consistently satisfied by construction when the data for $\R$ and $\tilde{\R}_{ab}$ matches the data of $R$ and $R_{ab} - (1/4) R g_{ab}$ in terms of the data of $\theta$ on $\mathcal{C}$, as we demand. Additionally, as usual, the eight constraints corresponding to the vanishing of initial data for the harmonic potential are tantamount to the existence of harmonic coordinates, which is guaranteed by the freedom to choose coordinates for the background metric components $\g_{\mu \nu}$, such that on the initial surface $F^\mu|_{\mathcal{C}} = 0$ and, as usual, that $\dot F^\mu|_{\mathcal{C}} = 0$ follows from $F^\mu|_{\mathcal{C}} = 0$ and the generalised Gauss and Hamilton constraints on $\mathcal{C}$.
\end{proof}

The validity of the above theorem relies on the existence of solutions to the constraint equations of semiclassical gravity discussed in sec. \eqref{subsec:Constraints}. It is therefore appropriate to give at least one example of initial data that satisfies these constraints. 

\begin{ex}[de Sitter-like data on the Cauchy surface] Consider that for the geometrical sector the initial data matches induced data from de Sitter spacetime, which is conformal to Minkowski spacetime with conformal factor $\ee^{2 \theta} = \alpha^2/\eta^2$. Thus, we have that the background metric is the flat metric
\begin{align}
\g_{ab} = -\dd \eta_a \otimes \dd \eta_b + \dd x_a \otimes \dd x_b + \dd y_a \otimes \dd y_b + \dd z_a \otimes \dd z_b
\end{align}
and on the Cauchy surface $\mathcal{C}$, defined at $\eta = \eta_0$, we have as data the real, smooth functions
\begin{align}
\theta_0 = \ln(\alpha/\eta_0), & \quad & \theta_1 = -1/\eta_0, & \quad & \theta_1 = 1/\eta_0^2, & \quad & \theta_1 = -2/\eta_0^3. 
\label{thetadataex}
\end{align} 

For the data of the two-point function we can use the correlation functions of a massless, conformally coupled Klein-Gordon field in the Bunch-Davies vacuum \cite{Bunch:1977sq},
\begin{subequations}
\begin{align}
G^+_{\varphi \varphi} & = \frac{1}{16 \pi \alpha^2} {}_2F_1\left(7/2,5/2;2;1 - \frac{ (\Delta \underline{x})^2 + \epsilon^2}{4 \eta_0^2}\right), \\
G^+_{\pi \varphi} & = \frac{35 \left( (\Delta \underline{x})^2 - 2 \ii \epsilon \eta_0 \right) \, _2F_1\left(7/2,9/2;3;1-  \frac{ (\Delta \underline{x})^2+ \epsilon^2}{4 \eta_0^2}\right)}{514 \pi  \alpha \eta_0^4}, \\
G^+_{\pi \varphi} & = \frac{35 \left( (\Delta \underline{x})^2 + 2 \ii \epsilon \eta_0 \right) \, _2F_1\left(7/2,9/2;3;1-  \frac{ (\Delta \underline{x})^2+ \epsilon^2}{4 \eta_0^2}\right)}{514 \pi  \alpha \eta_0^4}, \\
G^+_{\pi \pi} & = \frac{35}{8192 \pi \eta_0^8}  \left(-16 \eta_0^2 \left(2 \eta_0^2+(\Delta \underline{x})^2+\epsilon ^2\right) \, _2F_1\left(7/2,9/2;3;1- \frac{ (\Delta \underline{x})^2+ \epsilon^2}{4 \eta_0^2}\right) \right. \nonumber \\
& \left. +21 \left(\eta_0^2+(\Delta \underline{x})^2+(\epsilon -\ii \eta_0^2\right) \left(\eta_0^2+(\Delta \underline{x})^2+(\epsilon +\ii \eta_0)^2\right) \, _2F_1\left(\frac{9}{2},\frac{11}{2};4;1-\frac{ (\Delta \underline{x})^2+ \epsilon^2}{4 \eta_0^2}\right)\right).
\end{align}
\label{G+dataex}
\end{subequations}

(The detailed form of the data ${\sf G}^+_{\varphi \varphi}$, ${\sf G}^+_{\pi \varphi}$, ${\sf G}^+_{\varphi \pi}$ and ${\sf G}^+_{\pi \pi}$ can be obtained from \eqref{thetadataex} and \eqref{G+dataex}.)




Setting $\alpha = \frac{45 \pm \sqrt{15(135+4 G_{\rm N} \Lambda})}{30 \Lambda}$ one obtains a solution to the constraints of semiclassical gravity of sec. \ref{subsec:Constraints} and this constitutes a good set of initial data for semiclassical gravity for which theorem 
\ref{MainThm2} applies.
\end{ex}

\section{Final remarks}
\label{sec:final}


In this work we have been concerned studying solutions to semiclassical gravity with spacetime dependence. While results are already known in cases in which the solutions depend on a time parameter, notably in cosmological contexts, or where the solutions have non-trivial space-dependence, but are static, solutions that depend on both space and time had been lacking. To the best of our knowledge, the present work gives the first structural results in this context.
 
Theorem \ref{MainThm2} shows that the theory of semiclassical gravity for conformally static spacetimes is well posed with a conformally covariant Klein-Gordon field. This includes the well-posedness of semiclassical cosmology even in the case in which the scale factor contains anisotropies, provided that appropriate solutions for the constraint equations can be found. The task of solving the constraint equations in cases of interest might be feasible to specialists in numerical methods, such as numerical relativists. In the case of cosmology it might be also possible, and considerably simpler, to obtain results for anisotropic perturbations in cosmology using a combination of numerical and analytic techniques.  In any case, cosmological applications are an interesting avenue to explore in their own right in the present context. 

A natural extension of this work including the Maxwell field should yield a mathematically sound formalism for studying radiation-dominated cosmology.


\section*{Acknowledgments}

BAJ-A is supported by a CONACYT Postdoctoral Research Fellowship and acknowledges additional support from UNAM-DGAPA-PAPIIT grant IG100120. SM acknowledges the support CONACYT project A1S33440. BAJ-A and SM acknowledge the support of CONACYT project 140630. 

\appendix

\section{Expresions in the proof of theorem \ref{MainThm2}}
\label{app}

We detail the expressions appearing in the system \eqref{ExtendedSystem}. In the expressions below by the Riemann tensor we denote
\begin{align}
R_{abc}{}^d = {\sf R}_{abc}{}^d + 2 \delta^d{}_{[a} w_{|c|b]} - 2 \g^{de} \g_{c[a}  w_{|e|b]} + 2 v_{[a} \delta^d_{b]} v_c - 2 v_{[a} \g_{b]c} \g^{df} v_f - 2 \g_{c[a} \delta^d_{b]} \g^{ef} v_e v_f,
\label{RiemannConf}
\end{align}
which depends on the background structure and on $v_a$ and $w_{ab}$. It follows from \eqref{RiemannConf} that $\nabla_e^{(\g)}R_{abc}{}^d$ depends only on the background structure, $v_a$, $w_{ab}$, $\nabla_e^{(\g)}v_a$ and $\nabla_e^{(\g)} w_{ab}$, i.e., it contains no second-order derivatives of the dynamical variables. We have

\begin{align}
\mathscr{A} & = - \g^{ab} \left(2 \delta^c_{(a} v_{b)} - \g_{ab} \g^{cd} v_d \right) \mathscr{V}_c  - \ee^{2 \theta} \left( 2 \alpha  - 6 \beta + \frac{ G_{\rm N}}{90 \pi} \right)^{-1} \left[ -  2 \alpha \left( \tilde{\R}^{cd} + \frac{1}{4} \R \ee^{-2 \theta}{\sf g}^{cd} \right) R_{cda}{}^a   +  2 \beta \R^2 \right. \nonumber \\
 & \left. -\frac{ G_{\rm N}}{360 \pi}     R_{cdef} R^{cdef} + \left( 2 \alpha + \frac{ G_{\rm N}}{360 \pi} \right) \left( \tilde{\R}^{cd} + \frac{1}{4} \R \ee^{-2 \theta} {\sf g}^{cd} \right) \left( \tilde{\R}_{cd} + \frac{1}{4} \R \ee^{2 \theta} {\sf g}_{cd} \right)  -  \R  + 4 \Lambda \right], \\
\mathscr{A}_a & = \g^{cd} {\sf R}_{acd}{}^e \mathscr{V}_e + \nabla_a \mathscr{A}, \\
\mathscr{B}_{ab} & = \mathscr{D}_{ab} - \frac{\ee^{2 \theta}}{\alpha} \left[(\alpha - 2 \beta) \left(\nabla_a^{(\g)} \mathscr{V}_b - (2 \delta^c{}_{(a}v_{b)} - \g_{ab} \g^{cd} v_d) \mathscr{V}_c \right)   -  2 \alpha \left( \tilde{\R}^{cd} + \frac{1}{4} \R \ee^{-2 \theta} {\sf g}^{cd} \right) \left(R_{cdab} - \frac{1}{4} \ee^{2 \theta} {\sf g}_{ab} R_{cde}{}^e \right)   \nonumber \right.\\
& + 2 \beta \R \tilde{\R}_{ab}  +  \tilde{\R}_{ab} - \frac{1}{4} (\alpha - 2 \beta)    \ee^{2 \theta} {\sf g}_{ab} \left( 2 \alpha  - 6 \beta + \frac{ G_{\rm N}}{360 \pi} \right)^{-1} \left[ -  2 \alpha \left( \tilde{\R}^{cd} + \frac{1}{4} \R \ee^{- 2 \theta} {\sf g}^{cd} \right) R_{cda}{}^a -\frac{ G_{\rm N}}{360 \pi}     R_{cdef} R^{cdef} \right. \nonumber \\
 & \left. \left. + \left( 2 \alpha + \frac{ G_{\rm N}}{360 \pi} \right) \left( \tilde{\R}^{cd} + \frac{1}{4} \R \ee^{-2 \theta} {\sf g}^{cd} \right) \left( \tilde{\R}_{cd} + \frac{1}{4} \R \ee^{2 \theta} {\sf g}_{cd} \right)  -  \R  + 4 \Lambda \right] - 8 \pi G_{\rm N} \left(   [\mathcal{T}_{ab} w_\ell]- \frac{1}{4} {\sf g}_{ab} {\sf g}^{cd}[\mathcal{T}_{cd} w_\ell] \right) \right], \\
  \mathscr{C}_a & =   6 \g^{cd} {\sf R}_{acd}{}^b v_b + 6 \g^{bc} (v_b \nabla_a^{(\g)}  v_c + v_c \nabla_a^{(\g)}  v_b) + 2 \ee^{2 \theta} v_a \R + \ee^{2 \theta} \nabla_a^{(\g)} \R - \nabla_a^{(\g)} {\sf R}, \\
  \mathscr{C}_{ab} & =  6 \g^{cd} (v_e \nabla_c^{(\g)} {\sf R}_{bda}{}^e  +  {\sf R}_{bda}{}^e \nabla_c^{(\g)} v_e) + 6 \g^{cd} ({\sf R}_{bcd}{}^e\nabla_e^{(\g)} v_a + {\sf R}_{bca}{}^e \nabla_d^{(\g)} v_e )\nonumber \\
 & + \nabla_b^{(\g)} \Big[ 6 \g^{cd} {\sf R}_{acd}{}^e  v_e+ 6 \g^{cd} (v_d w_{ca} + v_c  w_{da}) + 2 \ee^{2 \theta} v_a \R + \ee^{2 \theta} \mathscr{V}_a - \nabla_a^{(\g)} {\sf R}\Big],
\end{align}
where $\mathscr{D}_{ab}$ can be written in terms of $C^c{}_{ab} := 2 \delta^c{}_{(a} v_{b)} - \g_{ab} \g^{cd} v_d$ 
as
\begin{align}
\mathscr{D}_{ab} 
& = - \g^{cd} C^e{}_{cd} \nabla_e^{(\g)} \tilde{\R}_{ab} - \g^{cd} C^e{}_{ca} \nabla_d^{(\g)} \tilde{\R}_{eb} - \g^{cd} C^e{}_{cb}  \nabla_d^{(\g)} \tilde{\R}_{ae} \nonumber \\
& -  \g^{cd} (\nabla_c^{(\g)} C^e{}_{da} + C^e{}_{cf} C^f{}_{da} - C^f{}_{cd} C^e{}_{fa} - C^f{}_{ca} C^e{}_{df}) \tilde{\R}_{eb}  - \g^{cd} C^e{}_{da} (\nabla_c^{(\g)} \tilde{\R}_{eb} - C^f{}_{ce} \tilde{\R}_{fb} - C^f{}_{cb} \tilde{\R}_{ef}) \nonumber \\
& -  \g^{cd} (\nabla_c^{(\g)} C^e{}_{db} + C^e{}_{cf} C^f{}_{db} - C^f{}_{cd} C^e{}_{fb} - C^f{}_{cb} C^e{}_{df}) \tilde{\R}_{ea}  - \g^{cd} C^e{}_{db} (\nabla_c^{(\g)} \tilde{\R}_{ea} - C^f{}_{ce} \tilde{\R}_{fa} - C^f{}_{ca} \tilde{\R}_{ef})
\end{align}
and setting $[\breve {\sf w}_{ab}](\x) := \lim_{\x' \to \x} \nabla_b^{(\g)} \nabla_a^{(\g)}  [{\sf G}^+(\x, \x') - {\sf H}(\x, \x')]$ (note the difference in the covariant derivative with respect to the definition of $[{\sf w}_{ab}]$ in eq. \eqref{sfwab}),
\begin{align}
[\mathcal{T}_{ab} w_\ell] & = \frac{\ee^{2\theta}}{2(2 \pi)^2}\left[  \frac{1}{3} \left(v_a v_b + w_{ab} - C^c{}_{ab} v_c\right) [{\sf w}]  + v_a [{\sf w}_{b}] + v_b [{\sf w}_{a}] -  [\breve {\sf w}_{ab}] - C^c{}_{ab}[{\sf v}_c]  +\frac{1}{3} \left(  - 2 v_a [{\sf w}]_{;b} - 2 v_b [{\sf w}]_{;a}  \right. \right. \nonumber \\
& \left. \left. + \nabla_b^{(\g)} \nabla_a^{(\g)} [{\sf w}] -  C^c{}_{ab} [{\sf w}]_{;c} \right) + \frac{1}{2} \g_{ab} \g^{cd}\left( \left(\frac{1}{3}v_c v_d - \frac{2}{3}(w_{cd} - C^c{}_{ab} v_c) \right) [{\sf w}]  - v_c [{\sf w}_{d}] - v_d [{\sf w}_{c}]+  [\breve {\sf w}_{cd}]- C^e{}_{cd}[{\sf w}_e] \right. \right. \nonumber \\
& \left. -\frac{1}{6} \left( - 2 v_c [{\sf w}]_{;d} - 2 v_d [{\sf w}]_{;c} + \nabla_d^{(\g)} \nabla_c^{(\g)} [{\sf w}] - C^e{}_{cd} [{\sf w}]_{;e} \right) \right) \left. + \frac{1}{6} \left( \tilde{\R}_{ab} - \frac{1}{4} \ee^{2 \theta} \g_{ab} \R \right)[{\sf w}] \right].
\end{align}

\end{document}